\documentclass{llncs}
\usepackage{a4wide}

\usepackage{amsmath,amssymb}
\usepackage{graphicx}
\usepackage{algorithmic}
\usepackage{algorithm}
\usepackage{graphicx}
\usepackage{color}
\usepackage{paralist}

\newcommand{\red}[1]{{#1}}

\newcommand{\onenorm}[1]{ \| #1 \|_1}
\newcommand{\define}{:=}
\newcommand{\pVCG}{p^{\mathit {VCG}}}

\def\RR{\mathbb R}

\def\EE{\mathbb E}

\def\cF{\mathcal F}

\def\cO{\mathcal O}

\def\bzero{\mathbf 0}

\newcommand{\hide}[1]{}
\newcommand{\raf}[1]{(\ref{#1})}

\newcommand{\cV}{\ensuremath{\mathcal{V}}}

\newcommand{\abs}[1]{\ensuremath{\left|#1\right|}}

\newcommand{\R}{\ensuremath{\mathbb R}}
\newcommand{\Z}{\ensuremath{\mathbb Z}}

\newcommand{\argmin}{\ensuremath{\mathrm{argmin}}}

\newcommand{\set}[2]{\{ #1 : #2 \}}     

\def\eps{\varepsilon}

\newcommand{\hx}{\ensuremath{\hat x}}

\newcommand{\bone}{\ensuremath{\boldsymbol{1}}}

\newcommand{\Q}{\mathcal{Q}}

\newcommand{\QI}{{\mathcal Q}_{\mathcal I}}

\newcommand{\ceil}[1]{\ensuremath{\lceil #1 \rceil}}

\newcommand{\argmax}{\operatorname{argmax}}

\newcommand{\bx}{\bar{x}}

\newcommand{\wx}{\hat{x}}

\newcommand{\Khaled}[1]{{#1}}

\title{Towards More Practical Linear Programming-based Techniques for Algorithmic Mechanism Design\thanks{A preliminary version of these results was presented at SAGT 2015~\cite{EMR:MechanismDesignSAGT}.}}
\author{
Khaled Elbassioni\inst{1} 
\and
Kurt Mehlhorn \inst{2} 
\and
Fahimeh Ramezani \inst{3} 
}

\institute{Masdar Institute of Science and Technology,
Abu Dhabi, UAE\and Max Planck Institute for Informatics,
Campus E1 4, 66123, Saarbrucken, Germany \and Department of Mathematics,
University of Isfahan,
Isfahan 81746-73441,
Iran
\tt{kelbassioni@masdar.ac.ae,
mehlhorn@mpi-inf.mpg.de, f.ramezani@sci.ui.ac.ir}
}
\begin{document}
\date{}
\maketitle

\begin{abstract}
R.~Lavi and C.~Swamy (FOCS 2005, J.~ACM 2011) introduced a general method for obtaining truthful-in-expectation mechanisms from linear programming based approximation algorithms. Due to the use of the Ellipsoid method, a direct implementation of the method is unlikely to be efficient in practice. We propose to use the much simpler and usually faster multiplicative weights update method instead. The simplification 
comes at the cost of slightly weaker approximation and truthfulness guarantees.
\end{abstract}

\section{Introduction}\label{intro}

\emph{Algorithmic mechanism design} is the art of designing and implementing the rules of a game to achieve a desired outcome from a set of possible outcomes. Each player (agent) has a valuation that assigns a value to each possible outcome. The desired outcome is the one that maximizes the sum of the valuations; this sum is usually called \emph{social welfare}. The players are assumed to be selfish: they report valuations to the mechanism, which may differ from the true valuations. Players may lie about their valuations in order to direct the mechanism into an outcome favorable to them. The mechanism computes an outcome and payments for the players.
The \emph{utility of a player} is her/his value of the outcome computed by the mechanism minus her/his payment charged by the mechanism. The agents are interested in optimizing their personal utility. Social welfare and personal utilities are determined with respect to the true valuations of the players, although they are not public knowledge. The purpose of the payments is to incentivize the players to report their true valuations. 
A mechanism is \emph{truthful} if reporting the truth is a best strategy for each player irrespective of the inputs provided by the other players. A mechanism is \emph{efficient} if the outcome and the payments can be computed in polynomial time. The \emph{underlying optimization problem} is the computation of an outcome maximizing social welfare given the valutions of the players.

If the underlying optimization problem can be efficiently solved to optimality, the celebrated VCG mechanism (see, e.g., \cite{NRTV07}) achieves
truthfulness, social welfare optimization, and polynomial running time. The computation of the outcome and the computation of the payments requires to solve the underlying optimization problem to optimality. 

Many optimization problems are NP-hard and hence are unlikely to have an exact algorithm with polynomial running time. However, it might be possible to solve the problem approximately in polynomial running time. 
 
An example is the combinatorial auction problem. There is a set of $m$ items to be sold to a set of $n$ players. The (reported) value of a set $S$ of items to the $i$-th player is $v_i(S)$ with $v_i(\emptyset) = 0$ and $v_i(S) \le v_i(T)$ whenever $S \subseteq T$. Let $x_{i,S}$ be a 0-1 variable indicating that set $S$ is given to player $i$. Then $\sum_{S} x_{i,S} \le 1$ for every player $i$ as at most one set can be given to $i$, and $\sum_{i} \sum_{S; j \in S} x_{i,S} \le 1$ for every item $j$ as any item can be given away only once. The social welfare is $\sum_{i,S} v_i(S) x_{i,S}$. The linear programming relaxation is obtained by replacing the integrality constraints for $x_{i,s}$ by $0 \le x_{i,S} \le 1$. Note that the number $d$ of variables is exponential in the number of items, namely $d = n 2^m$. The linear program is of the packing type, i.e., if $x$ is feasible and $y \le x$, then $y$ is feasible. For the combinatorial auction problem, $O(\sqrt{n})$-approximation algorithms exist and these algorithms also provide the corresponding integrality-gap-verifier (the definition is given below) with $\alpha = 1/\sqrt{n}$ (\cite{Briest-Krysta-Voecking,Kolliopoulos-Stein,Raghavan98}). 

For many integer linear programming problems, approximation algorithms are known that first solve the corresponding linear programming relaxation and then construct an integral solution either by rounding or by primal-dual methods. Lavi and Swamy (\cite{LS05,LS11}) showed that certain linear programming based approximation algorithms for the social welfare problem
can be turned  into randomized mechanisms that are truthful-in-expectation, i.e., reporting the truth maximizes  the expected utility of a player. The LS-mechanism is powerful (see~\cite{LS05,LS11,CEF10,HKV11} for applications), but unlikely to be efficient in practice because of its use of the Ellipsoid method. 
\emph{We show how to use the multiplicative weights update method instead. This results in simpler algorithms at the cost of somewhat weaker approximation and truthfulness guarantees.}

We next review the LS-mechanism. It applies to integer linear programming problems of the packing type for which the linear programming relaxation can be solved exactly and for which an $\alpha$-integrality gap verifier is available. More precisely:
\begin{enumerate}
\item Let $\mathcal{Q} \subseteq \R_{\ge 0}^d$ be a \emph{packing polytope}, i.e., $\mathcal{Q}$ is a bounded convex polytope contained in the non-negative orthant of $d$-dimensional space with the property that if $y \in \mathcal{Q}$ and $x \le y$ then $x \in \mathcal{Q}$. The linear programming problem for $\mathcal{Q}$ asks to find for a given $d$-dimensional vector $v$ a point $x^* = \argmax_{x \in \Q} v^T x$. 
\item We use $\mathcal{Q}_{\mathcal{I}} \define \mathcal{Q} \cap \Z^d$ for the set of integral points in $\mathcal{Q}$. The integer linear programming problem for $\mathcal{Q}_{\mathcal{I}}$ asks to find for a given $d$-dimensional vector $v$ a point $x^* = \argmax_{x \in \QI} v^T x$. We use $x^1$, $x^2$, \ldots, $x^j$, \ldots\ to denote the elements of $\QI$ and $\mathcal{N}$ for the index set of all elements in $\mathcal{Q}_{\mathcal{I}}$. 
\item An \emph{$\alpha$-integrality-gap-verifier} for $\mathcal{Q}_{\mathcal{I}}$ for some $\alpha \in (0,1]$ is an efficient algorithm that on input  $v \in \R^d$ and $x^* \in \mathcal{Q}$, returns an $x \in \mathcal{Q}_{\mathcal{I}}$ such that 
\[     v^T x \ge \alpha v^T x^*.\]
\end{enumerate}

The mechanism consists of three main steps:
\begin{enumerate}
\item  Let $v_i \in \R_{\ge 0}^d, 1\le i\le n,$ be the reported valuation of the $i$-th player and let $v = \sum_i v_i$ be the accumulated reported valuation. Solve the LP-relaxation, i.e., find a maximizer $x^* = \argmax_{x \in Q} v^T x$ for the social welfare of the fractional problem, and determine the VCG prices\footnote{$p_i= \sum_{j \not= i} v_j^T (\hx - x^*)$, where $\hx = \argmax_{x \in Q} \sum_{j \not= i} v_j^T x$.}  $p_1,\ldots,p_n$. The allocation $x^*$ and the VCG-prices are a truthful mechanism for the fractional problem. \smallskip
\item Write $\alpha\cdot x^*$ as a \emph{convex combination of integral solutions} in $\mathcal{Q}$, i.e., $\alpha \cdot x^* = \sum_{j \in \mathcal{N}} \lambda_j x^j$, $\lambda_j \ge 0$, $\sum_{j \in \mathcal{N}} \lambda_j = 1$, and $x^j \in \mathcal{Q}_{\mathcal{I}}$. This step requires the $\alpha$-integrality-gap-verifier. \smallskip
\item Pick the integral solution $x^j$ with probability $\lambda_j$, and charge the $i$-th player the price $p_i \cdot (v_i^T x^j/ v_{i}^T x^*)$. If $v_i^T x^* = 0$, charge zero. 
\end{enumerate}
The LS-mechanism approximates social welfare with factor $\alpha$ (is \emph{$\alpha$-socially efficient}) and guarantees {truthfulness-in-expectation}, i.e., it converts a truthful fractional mechanism into an $\alpha$-approximate truthful-in-expectation integral mechanism. 
With respect to practical applicability, steps 1 and 2 are the two major bottlenecks. Step 1 requires solving $n+1$ linear programs, one for the fractional solution and one for each price; an exact solution requires the use of the Ellipsoid method (see e.g. \cite{GLS88}), if the dimension is exponential. 
Furthermore, up to recently, the only method known to perform the decomposition in Step 2 is through the Ellipsoid method. An alternative method avoiding the use of the Ellipsoid method was recently given by Kraft, Fadaei, and Bichler~\cite{Kraft-Fadaei-Bichler}. We comment on their result in the next section.

\subsection{Our Results}
Our result concerns the design and analysis of a practical algorithm  for the LS-scheme. We first consider the case where the LP-relaxation of SWM (social welfare maximization)  in Step 1 of the LS-scheme  can be solved exactly and efficiently  and then our problem reduces to the design  of  a practical algorithm for Step 2.
Afterwards,  we consider the more general problem where only an FPTAS for the LP-relaxation is available.

\paragraph{Convex Decomposition.}
Over the past 15 years, simple and fast methods have been developed for solving packing and covering linear programs~ \cite{BI06,GK95,GargK07,K04,KY07,PST91,Y01} within an arbitrarily small error guarantee $\varepsilon$. These methods are based on the multiplicative weights update (MWU) method~\cite{AHK}, in which a very simple update rule is repeatedly performed until a near-optimal solution is obtained. We show how to replace the use of the Ellipsoid method in Step 2 by an approximation algorithm for covering linear programs. This result is the topic of Section~\ref{Fast Decomposition}. 

\begin{theorem}\label{t-mw} Let $\varepsilon > 0$ be arbitrary. 
Given a fractional point $x^* \in \mathcal{Q}$, and an $\alpha$-integrality-gap verifier for $\mathcal{Q}_{\mathcal{I}}$, we can find 
a convex decomposition 
\[   \frac{\alpha}{1+4\varepsilon}\cdot x^*=\sum_{j\in\mathcal{N}}\lambda_jx^j. \]
The convex decomposition has size (= number of nonzero $\lambda_j$) at most $s(1 + \ceil{\varepsilon^{-2} \ln s})$, where $s$ is the size of the support of $x^*$ (= number of nonzero components). The algorithm
makes at most $s \ceil{\varepsilon^{-2} \ln s}$ calls to the integrality-gap-verifier. 
\end{theorem}
Kraft, Fadaei, and Bichler~\cite{Kraft-Fadaei-Bichler} obtained a related result independently. However, their construction is less efficient in two aspects. First, it requires $O(s^2 \varepsilon^{-2})$ calls of the integrality-gap-verifyer. Second, the size of their convex decomposition might be as large as $O(s^3 \varepsilon^{-2})$. In the combinatorial auction problem, $s=n+m$.
Theorem~\ref{t-mw} together with Steps 1 and 3 of the LS scheme implies a mechanism that is truthful-in-expectation and has $(\alpha/(1 + 4\varepsilon))$-social efficiency. 

We leave it as an open problem whether the quadratic dependency of the size of the decomposition on $\varepsilon$ can be improved\footnote{We remark that recent progress \cite{ZO15,WRM15} on solving LPs of the packing/covering type has resulted in an almost {\it linear} dependence of the running time on $\frac{1}{\eps}$. However, the current methods do not work in the oracle model and hence cannot be directly applied in our setting.}.   

\paragraph{Approximately Truthful-in-Expectation Mechanism.}
We drop the assumption that the fractional SWM-problem can be solved optimally and assume instead that we have an FPTAS for it. We assume further that the problem is \emph{separable}, which means that the variables can be partitioned into disjoint groups, one for each player, such that the value of an allocation for a player depends only on the variables in his group, i.e, \[v_i(x)=v_i(x_i),\] where $x_i$ is  the 
set of variables associated\footnote{In the combinatorial auction problem, $x_i$ comprises all variables $x_{i,S}$. The value of an allocation $x$ for player $i$ is given by $\sum_S v_i(S) x_{i,S}$.} with player $i$. Formally,  any outcome $x \in\mathcal{Q}\subseteq \R^d$ can be written as $x = (x_1,\ldots,x_n)$ where $x_i \in \R^{d_i}$ and \mbox{$d = d_1 + \ldots + d_n$}. 
We further  assume that for each player $i\in[n]$, there is a \emph{dominating allocation}
 $u^i \in \mathcal{Q}$ that maximizes his value for every valuation $v_i$, i.e., 
\begin{equation}\label{def of ui}
  v_i(u^i) = \max_{z \in\mathcal{Q}}v_i(z),
\end{equation}
for every  $v_i\in \cV_i$, where $\cV_i$ denotes the possible
valuations of player $i$. For the case of a combinatorial auction, the allocation $u^i$ allocates all items to player $i$.  

\begin{theorem}\label{conversion} Let $\varepsilon_0 \in(0, 1/2]$. Define $\varepsilon=\Theta(\frac{\varepsilon_0^5}{n^4})$. Assuming that the fractional SWM-problem has an FPTAS, is separable, and has a dominant allocation for every player $i$, and that there is an $\alpha$-integrality gap verifier for $\QI$, there is a polynomial time
randomized integral mechanism with the following properties:
 \begin{compactenum}[\mbox{}\hspace{\parindent}(C1)]
\item No positive transfer, i.e., prices are nonnegative. \label{C1a}
\item Individually rational with probability $1 - \varepsilon_0$, i.e., i.e.,  the utility of any truth-telling player is non-negative with probability at least $1 - \varepsilon_0$.\label{C2a}
\item $(1 - \varepsilon_0)$-truthful-in-expectation, i.e., reporting the truth maximizes the expected utility of a player up to a factor $1 - \varepsilon_0$.\label{C3a}
\item $\gamma$-socially efficient, where $\gamma=\alpha(1 - \varepsilon)(1 - \varepsilon_0)/(1 +4 \varepsilon)$.\label{C4a} 
\end{compactenum}
\end{theorem}

Our mechanism is based on  constructing a randomized fractional mechanism with properties (C1) to (C3) and being $(1 - \varepsilon) (1-\varepsilon_0)$-socially efficient and then 
converting the mechanism into an integral mechanism with the properties above. The conversion is simple. Let us assume that  $x$ is a fractional allocation obtained from the fractional mechanism. 
We apply our convex decomposition technique and Step 3 of the Lavi-Swamy mechanism to obtain an integral randomized mechanism that satisfies  (C1) to (C4).We show this result in Section \ref{ApproxTruth}.
   
Our fractional mechanism refines the one given in~\cite{DRV11}, where the dependency of $\varepsilon$ on $n$ and $\varepsilon_0$ is as  $\varepsilon =\Theta({\varepsilon_0}/{n^9})$.  A recent experimental study of our mechanism on Display Ad Auctions \cite{EJ15} shows the applicability of our techniques in practice.

We leave it as an open problem whether the dependency of $\varepsilon$ on $\varepsilon_0$ and $n$ can be improved.

\paragraph{On the Existence of an FPTAS for the Fractional SWM-Problem.} We close the survey of our results with a comment on the existence of an FPTAS for the fractional SWM-problem. Consider a packing linear program
\[ \max c^{T}x\quad\text{subject to}\quad  Ax\leq b,~~~ x\geq 0, \]
where $A\in\RR_{\ge 0}^{m\times \red{n}}$ is an $m\times \red{n}$ matrix with non-negative entries and $c\in\RR_{>0}^n,$ $b\in\RR_{> 0}^m$ are positive vectors. We may assume that each column of $A$ contains a non-zero entry as otherwise the problem is trivially unbounded.
For every $\kappa\ge1$ and weight vector
$z\in \mathbb{R}^m_{\ge 0}$, 
let $\cO_\kappa(z)$ denote a $\kappa$-approximation oracle that 
 returns a $j$ such that
\begin{equation*} \frac{1}{c_j}\sum_{i=1}^{m}\frac{z_ia_{ij}}{b_i} \le \kappa \cdot \min_{j'\in[n]}\frac{1}{c_{j'}}\sum_{i=1}^{m}\frac{z_ia_{ij'}}{b_i}.\end{equation*}
Garg and K\"{o}nemann~\cite{GargK07} presented an algorithm that uses the oracle $\cO_\kappa$ to construct an 
approximation with a factor arbitrarily close to $1/\kappa$. For $\kappa = 1$, their algorithm is an FPTAS. 

What is the approximation oracle in case of the combinatorial auction problem?  In this problem, we have one constraint for each player and one constraint for each item. Let $y_i \ge 0$ be the weight for agent $i$ and $z_j \ge 0$ be the weight for item $j$. Then oracle  $\cO_1(y,z)$ requires to find the pair 
\[  (i,S) \define \argmin_{(k,T)} \frac{1}{v_{k}(T)} \left(y_k + \sum_{j \in T} z_j\right).\]
In other words, for each $k$, one needs to find the set $T$ which minimizes $(y_k + \sum_{j \in T} z_j)/v_k(T)$. If $y_k$ is interpreted as a fixed cost incurred by agent $k$ and $z_j$ as the cost of item $j$, then $T$ is the set that minimizes the ratio of cost relative to value. For a simple-minded bidder who is interested in the items in a subset $T_0$ and no other item, i.e., $v_k(T) = v_k(T_0)$ if $T_0 \subseteq T$ and $v_k(T) = 0$, otherwise, $T_0$ is the minimizer. Another simple case is additive valuations, i.e., $v_k(T) = \sum_{j \in T} a_j^k$, where $a_j^k \ge 0$ is the value of item $j$ for agent $k$. In this situation, $\frac{1}{v_{k}(T)} \left(y_k + \sum_{j \in T} z_j\right) \le \beta$ for a set $T$ and a positive real $\beta$ if and only if $\sum_{j \in T} (\beta a_j^k - z_j) \ge y_k$ and hence the minimal $\beta$ for which such a set $T$ exists is readily determined by binary search on $\beta$.

\section{A Fast Algorithm for Convex Decompositions}\label{Fast Decomposition}


Let $x^* \in \mathcal{Q}$ be arbitrary. Carr and Vempala~\cite{CV02} showed how to construct a convex combination of points in $\mathcal{Q}_{\mathcal{I}}$ dominating $\alpha x^*$ using a polynomial number of calls to an $\alpha$-integrality-gap-verifier for $\mathcal{Q}_{\mathcal{I}}$. Lavi and Swamy~\cite{LS11} modified the construction to get an exact convex decomposition $\alpha x^* = \sum_{i \in \mathcal{N}}\lambda_i x^i$ for the case of packing linear programs. 
The construction uses the Ellipsoid method. We show an approximate version that replaces the use of the Ellipsoid method by the multiplicative weights update (MWU) method.
For any $\varepsilon > 0$, we show how to obtain a convex decomposition of $\alpha x^*/(1 + \varepsilon)$. Let $s$ be the number of non-zero components of $x^*$. The size of the decomposition and the number of calls to the $\alpha$-integrality gap verifier are $O(s \varepsilon^{-2} \ln s)$. 

This section is structured as follows. We first review Khandekar's FPTAS for covering linear programs (Subsection~\ref{s-coverLP}). We then use it and the $\alpha$-integrality gap verifier to construct, on input $x^* \in \mathcal{Q}$, a dominating convex combination for $\alpha x^*/(1 + 4 \varepsilon)$ (Subsection~\ref{dominating convex combination}). In Subsection~\ref{exact convex decomposition}, we show how to convert a dominating convex combination into an exact convex decomposition. Finally, in Subsection~\ref{fast convex decomposition}, we put the pieces together. 

\subsection{Khandekar's Algorithm for Covering Linear Programs}
 \label{s-coverLP}
Consider a covering linear program:
\begin{align}
\label{cover}
\min c^{T}x \quad\text{subject to}\qquad Ax\geq b,~~ x\ge 0,
\end{align}
where $A\in\RR_{\ge 0}^{m\times n}$ is an $m\times n$ matrix with non-negative entries and $c\in\RR_{\ge 0}^n$ and  $b\in\RR_{\ge 0}^m$ are non-negative vectors.
 We assume the availability of a $\kappa$-\emph{approximation oracle} for some $\kappa \in (0,1]$. 
\begin{description}
\item[$\cO_\kappa(z)$:] Given $z\in \mathbb{R}^m_{\ge 0}$, the oracle finds a column $j$ of $A$ that maximizes $\frac{1}{c_j}\sum_{i=1}^{m}\frac{z_ia_{ij}}{b_i}$ within a factor of $\kappa$: 
\begin{equation*} \frac{1}{c_j}\sum_{i=1}^{m}\frac{z_ia_{ij}}{b_i} \ge \kappa \cdot \max_{j'\in[n]}\frac{1}{c_{j'}}\sum_{i=1}^{m}\frac{z_ia_{ij'}}{b_i} \end{equation*}
\end{description}
For an exact oracle $\kappa=1$, Khandekar~\cite{K04} gave an algorithm  which computes a feasible solution $\hat{x}$ to \raf{cover} such that $c^T\hat{x}\leq (1+4\varepsilon)z^*$ where  $z^*$ is the value of an optimal solution.  The algorithm makes $O(m\varepsilon^{-2}\log m)$ calls to the oracle, where $m$ is the number of rows in $A$. There are algorithms predating Khandekar's work, see, for example,~\cite[Chapter 4]{Koenemann-MS}

\begin{theorem}[Generalization of Khandekar's algorithm to arbitrary $\kappa \le 1$]\label{t1} Let $\varepsilon\in(0,\frac{1}{2}]$ and let $z^*$ be the value of an optimum solution to \raf{cover}.
Procedure Covering$(\cO_\kappa)$ (see Algorithm~\ref{Cover-LP} in Appendix I) terminates in at most $m\ceil{\varepsilon^{-2} \ln m}$ iterations with a feasible solution $\hat x$ of \raf{cover} of at most $m\ceil{\varepsilon^{-2} \ln m}$ positive components. At termination, it holds that
\begin{equation}\label{bd11}
c^T\hat x\le\frac{(1+4\varepsilon)}{\kappa}z^*.
\end{equation}  
\end{theorem}

For completeness, we give a proof of Khandekar's result in Appendix I.
The proof of Theorem \ref{t1} can be modified to give (see Appendix I):

\begin{corollary}\label{c1}
Suppose $b=\bone$, $c=\bone$, and we use the following oracle
$\cO'$ instead of $\cO$ in Algorithm \ref{Cover-LP}:

\begin{description}
\item[$\cO'(z)$:] Given $z\in \mathbb{R}^m_{\ge 0}$, such that $\bone^Tz=1$, the oracle finds a column $j$ of $A$ such that $z^TA \bone_j\ge 1$. 
\end{description}
Then the algorithm terminates in at most $m\ceil{\varepsilon^{-2} \ln m}$ iterations with a feasible solution $\hat x$ having at most $m\ceil{\varepsilon^{-2} \ln m}$ positive components, such that $\bone^T\hat x\le 1+4\varepsilon$. 
\end{corollary}
\subsection{Finding a Dominating Convex Combination}\label{dominating convex combination}
Recall that we use $\mathcal{N}$ to index the elements in $\mathcal{Q}_{\mathcal{I}}$. We assume the availability of an  $\alpha$-integrality-gap-verifier $\cF$ for $\mathcal{Q}_{\mathcal{I}}$. We will use the results of the preceding section and show how to obtain for any $x^*\in \mathcal{Q}$ and any positive $\varepsilon$ a convex composition of points in $\mathcal{Q}_{\mathcal{I}}$ that covers $\alpha x^*/(1 + 4 \varepsilon)$. Our algorithm requires $O(s \varepsilon^{-2} \ln s)$ calls  to the oracle, where $s$ is size of the support of $x^*$.
\begin{theorem}\label{Dominating Combination} Let $\varepsilon > 0$ be arbitrary. 
Given a fractional point $x^* \in \mathcal{Q}$ and an $\alpha$-integrality-gap verifier $\cF$ for $\mathcal{Q}_{\mathcal{I}}$, we can find a convex combination $\bar{x}$ of integral points in $\mathcal{Q}_{\mathcal{I}}$ such that 
\[         \frac{\alpha}{1+4\varepsilon}\cdot x^* \le \bar{x} = \sum_{i\in\mathcal{N}}\lambda_ix^i .\]
The convex decomposition has size at most $s \ceil{\varepsilon^{-2} \ln s}$, where $s$ is the number of positive entries of $x^*$.
The algorithm makes at most $s \ceil{\varepsilon^{-2} \ln s}$ calls to the integrality-gap verifier.
\end{theorem}
\begin{proof}
The task of finding the multipliers $\lambda_i$ is naturally formulated as a covering LP (\cite{CV02}), namely,
\begin{align}\label{ffs}
\min \quad & \sum_{i\in\mathcal{N}} \lambda_{i}\\
s.t.   &   \sum_{i\in\mathcal{N}} \lambda_{i}x^i_{j} \geq \alpha\cdot x_{j}^* \quad \text{for all $j$,}\nonumber\\
& \sum_{i\in\mathcal{N}} \lambda_i \geq 1,~~\lambda_{i} \geq 0.\nonumber\\
&\lambda_{i} \geq 0.\nonumber
\end{align}
Clearly, we can restrict our attention to the $j \in S^+ \define \set{j}{x^*_j > 0}$ and rewrite the constraint for $j \in S^+$ as $ \sum_{i\in\mathcal{N}} \lambda_{i}x^i_{j}/ (\alpha\cdot x_{j}^*) \ge 1$. For simplicity of notation, we assume $S^+ = [1 .. s]$. We thus have a covering linear program as in (\ref{cover}) with $m \define s+1$ constraints, $n \define \abs{\mathcal{N}}$ variables $\lambda_i$, right-hand side $b \define \bone$, cost vector $c \define \bone$, and constraint matrix $A = (a_{j,i})$ (note that we use $j$ for the row index and $i$ for the column index), where
  \[
 a_{j,i} \define \left\{
 \begin{array}{l l}
x^i_j/(\alpha x^*_j) &  \quad 1\leq j\leq s, i\in\mathcal{N}\\
1 & \quad j=s+1, i\in\mathcal{N}\\
 \end{array} \right.
 \]
 Thus we can apply Corollary~\ref{c1} of Section~\ref{s-coverLP}, provided we can efficiently implement the required oracle $\cO'$. We do so using ${\cF}$. 

Oracle $\cO'$ has is given a $\tilde{z})$ such that $1^T\tilde{z}=1$. Let us conveniently write $\tilde{z}=(w,z)$, where $w\in \mathbb{R}_{\ge 0}^{s}$, $z\in\RR_{\ge 0}$, and $\sum^{j=1}_{j = s}w_j+z=1$.  
Oracle $\cO'$ needs to find a column $i$ such that $\tilde{z}^TA\bone_i\geq 1$. In our case $\tilde{z}^TA\bone_i=\sum_{j=1}^s w_j x^i_j/\alpha x^*_j+z$, and we need to find a column $i$ for which this expression is at least one. Since $z$ does not depend on $i$, we concentrate on the first term. Define 
  \[
 V_j:= \left\{
 \begin{array}{l l}
\frac{w_j}{\alpha x^*_j}&  \quad \text{for $j\in S^+$}\\
0  & \quad \text{otherwise}.\\
 \end{array} \right.
 \]

\noindent
Call algorithm ${\cF}$ with $x^*\in \mathcal{Q}$ and $V:=(V_1,\ldots )$. ${\cF}$ returns an integer solution $x^i\in \mathcal{Q_I}$ such that
\[ \sum_{j\in S^+}\frac{w_j}{\alpha x_j^*}x^i_j = V^T x^i \ge \alpha \cdot V^Tx^*=\sum_{j\in S^+} w_j,\]
and hence, 
\[ \sum_{j\in S^+} \frac{w_j}{\alpha x_j^*}x^i_j+z \ge\sum_{j\in S^+}w_j+z=1.\]
Thus $i$ is the desired column of $A$.

It follows by Corollary \ref{c1} that Algorithm \ref{Cover-LP} finds a feasible solution $\lambda' \in\RR_{\ge 0}^{|\mathcal{N}|}$ to the covering LP \raf{ffs}, and a set $\mathcal{Q}_\mathcal{I}'\subseteq\mathcal{Q}_\mathcal{I}$ of vectors (returned by $\mathcal{F}$), such that $\lambda'_i>0$ only for $i\in\mathcal{N}'$, where $\mathcal{N}'$ is the index set returned by  oracle $\mathcal{O}'$ and  $|\mathcal{N}' |\le s \ceil{\varepsilon^{-2} \ln s}$. Also $ \Lambda \define \sum_{i\in\mathcal{N}'}\lambda'_i\le (1+4\varepsilon)$. Scaling $\lambda'_i$ by $\Lambda$, we obtain a set of 
 multipliers $\{\lambda_i=\lambda_i'/\Lambda:~ i\in\mathcal{N}'\}$, such that  $\sum_{i\in\mathcal{N}'}\lambda_i =1$ and
\begin{equation*}
\sum_{i\in\mathcal{N}'}\lambda_i x^i\geq \frac{\alpha}{1+4\varepsilon} x^*.
\end{equation*}
We may assume $x^i_j=0$ for all $j\notin S^+$ whenever $\lambda_i > 0$; otherwise simply replace $x^i$ by a vector in which all components not in $S^+$ are set to zero. By the packing property this is possible. 
\qed
\end{proof}
\subsection{From Dominating Convex Combination to Exact Convex Decomposition}\label{exact convex decomposition}

We will show how to turn a dominating convex combination into an exact decomposition. The construction is general and uses only the packing property. Such a construction seems to have been observed in~\cite{LS05}, but was not made  explicit. Kraft, Fadaei, and Bichler~\cite{Kraft-Fadaei-Bichler} describe an alternative construction. Their construction may increase the size of the convex decomposition (= number of non-zero $\lambda_i$) by a multiplicative factor $s$ and an additive factor $s^2$. In contrast, our construction increases the size only by an additive factor $s$.

\begin{algorithm}[t]
\caption{Changing a dominating convex decomposition into an exact decomposition}
\label{dominating to exact decomposition}
\begin{algorithmic}[1]
\REQUIRE A packing convex set $\mathcal{Q}$ and point $x^* \in \mathcal{Q}$ and a convex combination $\sum_{i \in \mathcal{N}} \lambda_i x^i$ of integral points in $\mathcal{Q}_\mathcal{I}$ dominating $x^*$. 
\ENSURE A convex decomposition $x^* = \sum_{i \in \mathcal{N}'} \lambda_i x^i$ with $x^i \in \mathcal{Q}_\mathcal{I}$. 
\WHILE{$\Delta_j \define  \sum_{i \in \mathcal{N}} \lambda_i x^i -x^*_j > 0$ for some $j$}
\STATE let $i$ be such that $\lambda_i x^i_j > 0$ and $\Delta_j > 0$ for some $j$. 
\IF{there is a $j$ such that $\lambda_i x^i_j > 0$ and $\sum_{h \in \mathcal{N}} \lambda_h x^h - \lambda_i \bone_j \ge x^*$}
\STATE replace $x^i$ by $x^i - \bone_j$. 
\ELSE \STATE Among the indices $j$ with $x^i_j > 0$ and $\Delta_j > 0$, let $k$ minimize $\Delta_k/x^i_k$. 
\STATE let $y$ be such that $y_j = x^i_j$, if $\Delta_j = 0$, and $y_j = 0$, if $\Delta_j > 0$. 
\STATE change the lefthand side of \raf{dominance} as follows: replace $\lambda_i$ by $\lambda_i - {\Delta_k}/{x^i_k}$ and increase the coefficient of $y$ by $\Delta_k/x^i_k$. 
\ENDIF
\ENDWHILE
\end{algorithmic}
\end{algorithm}

\begin{theorem}\label{Dominating to Exact} Let $x^* \in \mathcal{Q}$ be dominated by a convex combination $\sum_{i \in \mathcal{N}} \lambda_i x^i$ of integral points in $\mathcal{Q}_{\mathcal{I}}$, i.e.,
\begin{equation} \sum_{i \in \mathcal{N}} \lambda_i x^i \ge x^*. \label{dominance} \end{equation}
Then Algorithm~\ref{dominating to exact decomposition} achieves equality in \raf{dominance}. It increases the size of the convex combination by at most $s$, where $s$ is the number of positive components of $x^*$. 
\end{theorem}
\begin{proof} Let $S^+ = \set{j}{x^*_j > 0}$. We may assume $x^i_j = 0$ for all $j \not\in S^+$ and all $i \in {\mathcal N}$ with $\lambda_i > 0$. 

For $j \in S^+$, let $\Delta_j =  \sum_{i\in\mathcal{N}}\lambda_i x^i_j - x^*_j $ be the gap in the $j$-th component. 
If $\Delta_j = 0$ for all $j\in S^+$, we are done. Otherwise, choose $j$ and $i \in \mathcal{N}$ such that $\Delta_{j} > 0$ and $\lambda_i x^i_{j} > 0$. 

Let $\bone_j$ be the $j$-th unit vector. If, for some $j$ with $x^i_j > 0$ and $\Delta_j > 0$, replacing $x^i$ by $x^i - \bone_j$ maintains feasibility,  i.e., satisfies constraint \raf{dominance}, we perform this replacement. 
Since $x^i$ is an integer vector in $\mathcal{Q}_{\mathcal{I}}$, the vector $x^i-\bone_j$ is nonnegative and, by the packing property, in $\mathcal{Q}_{\mathcal{I}}$.  The replacement decreases $\Delta_j$ by $\lambda_i$ and does not increase the number of nonzero $\lambda_i$. 

Otherwise, $\Delta_j < \lambda_i$ for all $j$ with $\Delta_j > 0$ and $x^i_j > 0$. Since $x^i$ is integral, we also have $\Delta_j \le \lambda_i x^i_j$ for all such $j$. Among the indices $j$ with $\Delta_j > 0$ and $x^i_j > 0$, let $k$ minimize $\Delta_k/x^i_k$. Let $y$ be such that $y_j = x^i_j$ if $\Delta_j = 0$ and $y_j = 0$ if $\Delta_j > 0$. Then $y \in \mathcal{Q}_{\mathcal I}$ since $\mathcal{Q}$ is a packing polytope. In the convex combination, replace
\[  \lambda_i x^i  \quad\text{by}\quad (\lambda_i - \frac{\Delta_k}{x^i_k})\cdot x^i + \frac{\Delta_k}{x^i_k} \cdot y.\]
Notice that $\lambda_i - \frac{\Delta_k}{x^i_k} \ge 0$. Let $\Delta_j'$ be the new gaps. Then clearly $\Delta_j' = \Delta_j$, if $\Delta_j = 0$. Consider any $j$ with $\Delta_j > 0$. Then
\[ \Delta'_j = \Delta_j - \frac{\Delta_k}{x^i_k} \cdot x^i_j = \begin{cases} 0  & \text{if $j = k$}\\
                                                                                                                   \ge  (\Delta_j - \frac{\Delta_j}{x^i_j})\cdot x^i_j = 0 &\text{if $j \not= k$.}\end{cases}\]
The inequality in the second case holds since $\Delta_k/x^i_k \le \Delta_j/x^i_j$. 
We have decreased the number of nonzero $\Delta_j$ by one at the cost one additional nonzero $\lambda_i$. 
Thus the total number of vectors added to the convex decomposition is at most $s$.
\qed
\end{proof}
\subsection{Fast Convex Decomposition}\label{fast convex decomposition}
We are now ready to prove Theorem~\ref{t-mw}.

\medskip
\noindent{\it Proof of Theorem~\ref{t-mw}}.
Theorem~\ref{Dominating Combination} yields a convex combination of integer points of $\mathcal{Q}_I$ dominating $\alpha x^*/(1 + 4 \varepsilon)$. The convex decomposition has size at most $s \ceil{\epsilon^{-2} \ln s }$, where $s$ is the number of positive entries
of $x^*$. The algorithm makes at most $s \ceil{\epsilon^{-2} \ln s }$ calls to the integrality-gap verifier.
Theorem~\ref{Dominating to Exact} turns this dominating convex combination into an exact combination. It adds up to $s$ additional vectors to the convex combination. \qed

\section{Approximatly Truthful-in-Expectation  Mechanisms}\label{ApproxTruth}

The goal of this section is to derive an approximate VCG-mechanism. We do not longer assume that the fractional SWM-problem can be solved exactly, but instead assume that we have an FPTAS for it. We will first design a randomized fractional algorithm (Theorem~\ref{fractional mechanism} in Subsection~\ref{Approximately Truthful Fractional Mechanism}) and then convert the fractional mechanism into an integral mechanism and prove Theorem~\ref{conversion} in Subsection~\ref{app:3}.

\subsection{Approximately  Truthful-in-Expectation Fractional Mechanisms}\label{Approximately Truthful Fractional Mechanism}

\begin{theorem}\label{fractional mechanism}\label{approximately truthful mechanism}
Let $\varepsilon_0 \in(0, 1/2]$. Define $\varepsilon=\Theta(\frac{\varepsilon_0^5}{n^4})$. Assuming that the fractional SWM-problem has an FPTAS, is separable, and has a dominant allocation for every player $i$, there is a polynomial time
randomized fractional mechanism (Algorithm~\ref{rel-TIE}) with the following properties:
 \begin{compactenum}[\mbox{}\hspace{\parindent}(D1)]
\item No positive transfer, i.e., prices are nonnegative. \label{D1}
\item Individually rational with probability $1 - \varepsilon_0$, i.e., the utility of any truth-telling player is non-negative with probability at least $1 - \varepsilon_0$.\label{D2}
\item $(1 - \varepsilon_0)$-truthful-in-expectation, i.e., reporting the truth maximizes the expected utility of a player up to a factor $1 - \varepsilon_0$.\label{D3}
\item $\gamma$-socially efficient, where $\gamma=(1 - \varepsilon)(1 - \varepsilon_0)$. \label{D4} 
\end{compactenum}
\end{theorem}

In order to present  Algorithm \ref{rel-TIE} and prove Theorem~\ref{fractional mechanism}, we introduce some notation and prove some preliminary Lemmas.
Let 
\begin{equation}\label{def of beta}      L_i \define  \sum_{j \ne i} v_j(u^j) \quad\text{and}\quad  \beta_i \define \varepsilon L_i. \end{equation} 
Note that $L_i$ does not depend on the valuation of player $i$. Let $\mathcal{A}$ be an $\varepsilon$-approximation algorithm  for the LP relaxation of SWM. Note that $\mathcal{A}$ is polynomial time since the running time of an FPTAS is polynomial in $\frac{1}{\varepsilon}$. We use $\mathcal{A}(v)$ to denote the outcome of $\mathcal{A}$ on input $v$; $\mathcal{A}(v)$ is a fractional allocation in $Q$. In the following, we will apply $\mathcal{A}$ to different valuations which we denote by $v = (v_i,v_{-i})$, $\bar{v} =  (\bar{v}_i, v_{-i})$, and $v' = (\bzero, v_{-i})$. Here $v_i$ is the reported valuation of player $i$, $\bar{v}_i$ is his true valuation and $v'_i=\bzero$. We denote the allocation returned by $\mathcal{A}$ on input $v$ (resp., $\bar{v}$, $v'$) by $x$ (resp., $\bx$, $x'$). Note that $x$, $\bx$, $x'$ are  fractional allocations.

We first bound the maximal change in social welfare induced by a change of the valuation of the $i$-th player. 
%
\begin{lemma}\label{easy} Let $\varepsilon \ge 0$ and let $\mathcal{A}$ be an $\varepsilon$-approximation algorithm which returns allocation $x$ on input vector $v$.  Let $\wx\in \mathcal{Q}$ be an arbitrary point, then 
\begin{equation} v(x) \ge v(\wx)-\beta_i -\varepsilon\cdot v_i(\wx)\label{cond}\end{equation}
for every $i$. 
\end{lemma}
\begin{proof} We have 
\begin{align*}
v(x)&\ge(1-\varepsilon)\max_{x \in \mathcal{Q}} v(x) \\
&\ge 
(1-\varepsilon)v(\wx)\nonumber \\ 
&=v(\wx) - \varepsilon\cdot \sum_{j\neq i}v_j(\wx)-\varepsilon\cdot v_i(\wx)\nonumber \\ 
&\ge v(\wx)-\beta_i-\varepsilon\cdot v_i(\wx),
\end{align*}
where the first inequality follows from the fact that $\mathcal{A}$ is an $\varepsilon$-approximation algorithm, and the last inequality follows from $\varepsilon\sum_{j\neq i}v_j(\wx) \leq\varepsilon\sum_{j\neq i} v_j(u^j)=\beta_i$.
\qed
\end{proof}
We use the following payment rule: 
\begin{align}\label{LS}
p_i(v):= \max\{p_i^{\mathit{VCG}}(v)-\beta_i,0\} 
\end{align}
{where} \[p_i^{\mathit{VCG}}(v):=v_{-i}(x')-v_{-i}(x).\] $v_{-i}(x)=\sum_{j\neq i}v_{j}(x), x ={\mathcal{A}}(v)$ and $x' ={\mathcal{A}}(0,v_{-i})$.  Observe the similarity in the definition of $p_i^{\mathit{VCG}}(v)$ to the VCG payment rule. In both cases, the payment is defined as the difference of the total value of two allocations to the 
players different from $i$. The first allocation ignores the influence of player $i$ ($x' ={\mathcal{A}}(0,v_{-i})$) and the second allocation takes it into account ($x ={\mathcal{A}}(v)$). The difference to the VCG rule is that $x'$ and $x$ are not true maximizers but are computed by an $\varepsilon$-approximation algorithm.

\begin{algorithm}[t]
	\caption{The mechanism $M$ of Theorem~\ref{approximately truthful mechanism}.
		The vectors $u^i$ are defined as in \raf{def of ui} and the quantities $L_i$ are defined in \raf{def of beta}. 
		The definitions of $q_0, q_j,$ active and inactive player  are given in the proof of Theorem \ref{approximately truthful mechanism}.}
	\label{rel-TIE}
	\begin{algorithmic}[1]
		\REQUIRE A valuation vector $v$, a packing convex set $\mathcal{Q}$  and an $\varepsilon$-approximation algorithm, where $\varepsilon$ is as Theorem \ref{approximately truthful mechanism}.

		\ENSURE An allocation $x\in\mathcal{Q}$ and a payment $p\in\RR^n$ satisfying (D1) to (D4).

		\STATE Choose an index $j\in \{0,1,\ldots,n\}$, where $0$ is chosen with probability $q_0$ and $j \in \{1,\ldots,n\}$ is chosen with probability $q_j= (1 - q_0)/n$. 
		
		\IF{$j=0$}
		\STATE Use $\varepsilon$-approximation algorithm $\mathcal{A}$ to compute an allocation $x = (x_1,\ldots ,x_n)\in\mathcal{Q}$ 
		and compute payments with payment rule \raf{LS}. For all inactive $i$, change $x_i$ and $p_i$ to zero.
		\ELSE
		\STATE For  every $1\leq i\leq n$, set
		\[
		\left\{
		\begin{array}{l l}
		x_i=u^i, p_i= \eta' L_i& \quad \text{if $i=j$ and $i$ is active, }\\
		x_i=u^i,  p_i=0 &\quad \text{if $i=j$ and $i$ is inactive,}\\
		x_i=0,  p_i=0 &\quad \text{if $i\neq j$.}\\
		\end{array} \right.
		\]
		\ENDIF
		
		\RETURN $(x,p)$
	\end{algorithmic}
\end{algorithm}

Let $U_i(v)=\bar{v}_i(x)-p_i(v)$ be the utility of player $i$ for bid vector $v$.  Note that the value of the allocation $x=\mathcal{A}(v)$ is evaluated with the true valuation $\bar{v}_i$ of player $i$. Let $U_i(\bar{v}) = \bar{v}_i(\bx)-p_i(\bar{v})$ be the utility of player $i$ for valuation vector $\bar{v}=(\bar{v}_i, v_{-i})$. 

 \begin{lemma}\label{beta-VCG} Let $\varepsilon \ge 0$ and let $\mathcal{A}$ be an $\varepsilon$-approximation algorithm.  Let $M_0$ be the mechanism with allocation function $\mathcal{A}(v)$
and the payment rule \raf{LS}.  $M_0$   is an individually rational mechanism with no positive transfer, such that for all $i$, 
\begin{equation}\label{ineq}
U_i(\bar{v})\ge U_i(v)-\varepsilon\cdot\bar v_i(x)-3\beta_i.
\end{equation}
\end{lemma}
\begin{proof}
By definition, $p_i(v)\ge 0$ for all $v$ and all $x$; so the mechanism has no positive transfer. 
We next address individual rationality.  Assume $p_i(\bar v)=\pVCG_i(\bar{v})-\beta_i>0$, as otherwise $U_i(\bar v)\ge 0$.  We have 
\begin{align*}
U_i(\bar{v})&=\bar{v}_i(\bx)-p_i(\bar{v})  \\
&=\bar{v}_i(\bx)- \pVCG_i(\bar{v})+\beta_i\\
&=\bar{v}_i(\bx)  + \bar{v}_{-i}(\bx) - \bar{v}_{-i}(x') +\beta_i \\
&=\bar{v}(\bx)   - \bar{v}(x') + \bar{v}_i(x') +\beta_i \\
&\ge(1-\varepsilon)\bar{v}_{i}(x')\ge 0,
\end{align*}
where the first inequality follows from Lemma \ref{easy} with $v = \bar{v}$ and $\wx = x'$. \smallskip

Finally, we prove \raf{ineq}. We have 
$v'(x')=v_{-i}(x')$, $v'(x)=v_{-i}(x)$, and $v'_{i}(x)=0$. Thus, 
\begin{align*}
\pVCG_i(v) & =v_{-i}(x')-v_{-i}(x)=v'(x')-v'(x)+\varepsilon\cdot v'_i(x)  \nonumber
\end{align*}
Applying Lemma  \ref{easy} for $v = v'$ and $\wx = x$, we obtain
\begin{align*}
v'(x')-v'(x)+\varepsilon\cdot v'_i(x) \ge - \beta_i
\end{align*}
Therefore, 
\begin{align}\label{ass}
\pVCG_i(v)+\beta_i\ge 0. 
\end{align}
 To see \raf{ineq}, we consider two cases:\smallskip

 \noindent\emph{Case 1:} $p_i(v)=0$. Then using \raf{ass}
 \begin{align*}  
 U_i(\bar{v}) = \bar{v}_i(\bx)-0\ge \bar{v}_i(\bx)-\pVCG_i(\bar{v})-\beta_i.
  \end{align*}

\noindent\emph{Case 2:} $p_i(v) = \pVCG_i(v)-\beta_i$. 
 \begin{align*}  
 U_i(\bar{v}) = \bar{v}_i(\bx)-p_i(\bar{v})=\bar{v}_i(\bx)-\pVCG_i(\bar{v})+\beta_i\ge \bar{v}_i(\bx)-\pVCG_i(\bar{v})-\beta_i,
  \end{align*}
  where the last inequality follows from $\beta_i\ge 0$. Therefore, in both cases we have:
 \[U_i(\bar{v}) \ge\bar{v}_i(\bx)-\pVCG_i(\bar{v})-\beta_i.\]
Now by using the definition of $\pVCG_i$ and Lemma \ref{easy}, we get  
\begin{align*}
U_i(\bar{v})& \ge\bar{v}_i(\bx)-\pVCG_i(\bar{v})-\beta_i\\
&=\bar{v}_i(\bx)+\bar{v}_{-i}(\bx)-\bar{v}_{-i}(x')-\beta_i\\
&=\bar{v}(\bx)-\bar{v}_{-i}(x')-\beta_i\\
&\ge\bar{v}(x)-\beta_i-\varepsilon\bar{v}_i(x)-\bar{v}_{-i}(x')-\beta_i\\
&=\bar{v}_i(x)-\pVCG_i(v)-\varepsilon \bar{v}_i(x)-2\beta_i\\
 &\ge \bar{v}_i(x)-p_i(v)-\beta_i-\varepsilon \bar{v}_i(x)-2\beta_i\\
 &=U_i(v)-\varepsilon \bar{v}_i(x)-3\beta_i.
\end{align*}
\qed
\end{proof}
 In what follows we  prove Theorem  \ref{approximately truthful mechanism}.

\medskip
\noindent{\it Proof of Theorem \ref{approximately truthful mechanism}}.
Define  
$q_0 = (1 - \frac{\varepsilon_0}{n})^n$, $\bar{\varepsilon} = \varepsilon_0/2$, and $q_j = (1 - q_0)/n$ for $1 \le j \le n$. Let  $\eta= \bar{\varepsilon}(1 - q_0)^2/n^3$, $\eta' = \eta/q_j$, and 
$\varepsilon = \eta \bar{\varepsilon} (1 - q_0)/(8 n)$.  Then using\footnote{Let $f(x)=(1-{x}/{n})^n - 1+x$. Then $f'(x) = (1 - x/n)^{n-1} +1 \ge 0$. Hence, for $0 \le x \le 1$ and $n \ge 1$, the function is increasing and $f(x)\ge f(0)=0$.}  $q_0= (1 -\frac{\varepsilon_0}{n})^n\ge 1-\varepsilon_0$ and $q_0= (1 -\frac{\varepsilon_0}{n})^n\le 1-\varepsilon_0/2$, we get
\[ \frac{ \varepsilon_0^5}{128 n^4}=\frac{ \bar{\varepsilon}^2 (\varepsilon_0/2)^3}{8 n^4}\le\varepsilon = \eta \bar{\varepsilon} (1 - q_0)/(8 n)=\frac{ \bar{\varepsilon}^2 (1 - q_0)^3}{8 n^4} \le \frac{ \bar{\varepsilon}^2 \varepsilon_0^3}{8 n^4}=\frac{\varepsilon_0^5}{16 n^4},\]
as stated in the Theorem. 
Let $U_i(v) = \bar{v}_i(x)-p_i(v)$  be the utility of player $i$ obtained by the mechanism $M_0$ of Lemma \ref{beta-VCG}. \Khaled{Let further $\widehat U_i(v) = v_i(x)-p_i(v)$}. Following \cite{DRV11}, we call player $i$ {\it {active}} if the following two conditions hold:
\begin{align}\label{active1}
\Khaled{\widehat U_{i}(v)}+\frac{\bar{\varepsilon} q_i}{q_0}v_i(u^i)&\ge\frac{q_i}{q_0} \eta'  L_i, \\
v_i(u^i)&\ge \eta L_i. \label{active2}
\end{align}
\Khaled{Note that these conditions do not depend on the true valuation $\bar v$.}
We denote by $T = T(v)$ the set of active players when the valuation is $v = (v_1,\ldots,v_n)$. Note that $L_i$ does not depend on $v_i$. Thus when we refer to conditions \raf{active1} and \raf{active2} for $\bar{v}$, we replace $v$ and $x$ by $\bar{v}$ and $\bx$ on the left side and keep the right side unchanged. 
Non-negativity of payments is immediate from the definition of  mechanism $M$ and Lemma~\ref{beta-VCG}. Moreover, the utility of a truth-telling bidder $i$ can be negative only if he/she is allocated in step 5, i.e., at most with probability $q_i$. It follows that the mechanism is individually rational with probability at least $1 - \sum_{i=1}^n q_i = q_0= (1 -\frac{\varepsilon_0}{n})^n\ge 1-\varepsilon_0$. 

Now we address truthfulness. Let us denote the expected utility  of player $i$ obtained from  the mechanism  in Algorithm~\ref{rel-TIE} on input $v\in\cV$ by $\EE[U'_i(v)]$.  Assume $j=0$ in Algorithm ~\ref{rel-TIE}. We run $\varepsilon$-approximation algorithm $\mathcal{A}$ on $v$ to compute  allocation $x=(x_1,\ldots, x_n)$. Then we change $x_i$ and $p_i$ to zero for all 
inactive $i$.  Let $\widetilde x$ be the allocation obtained in this way. The value for  player $i$ is $v_i(\widetilde x)$. When the $i$-th player is active, this value is equal to $v_i(x)$ because $v_i$ depends only on the valuation in the $i$-th group  (separability property).  Therefore in this case his utility  is $U_i(v)$. So we have that
\begin{align}\label{exp-util}
\EE[U'_i(v)]=\left\{
\begin{array}{ll}
q_0 \cdot U_i(v) + q_i (\bar{v}_i(u^i)-\eta' L_i) &\text{ if $i\in T(v)$,}\\
q_i \bar{v}_i(u^i) & \text{ if $i \not\in T(v)$}.
\end{array}
\right.
\end{align}

We first observe
\begin{equation}\label{min-util}
\EE[U'_i(\bar{v})]\ge(1-\bar{\varepsilon})q_i \cdot \bar{v}_i(u^i).
\end{equation}
Indeed, the inequality is trivially satisfied if $i\not\in T(\bar{v})$. On the other hand, if $i\in T(\bar{v})$, then \raf{active1} implies \Khaled{$U_i(\bar{v})=\widehat U_i(\bar v)\ge\frac{q_i}{q_0}\left(\eta' L_i-\bar{\varepsilon}\bar{v}_i(u^i)\right)$},  therefore 
\begin{align*}
\EE[U'_i(\bar{v})]
&=q_0 \cdot U_i(\bar{v}) + q_i (\bar{v}_i(u^i)-\eta' L_i)\\
&\ge q_0 \cdot \frac{q_i}{q_0}\left(\eta' L_i-\bar{\varepsilon}\bar{v}_i(u^i)\right)+ q_i (\bar{v}_i(u^i)-\eta' L_i) \\
&=(1-\bar{\varepsilon})q_i \cdot \bar{v}_i(u^i).
\end{align*}
We now consider four cases:

\smallskip

\noindent\emph{Case 1:} $i\in T(\bar v)\cap T(v)$. Note that  \raf{active2} for $\bar{v}$ implies $\beta_i = \varepsilon L_i  \le \frac{\varepsilon \bar{v}_i(u^i)}{\eta}$. Thus, by   Lemma~\ref{beta-VCG}, and using assumption \raf{def of ui} that $\bar{v}_i(x)\le\bar{v}_i(u^i)$, we have 
\begin{equation}\label{active2 for bv}
U_i(\bar{v})\ge U_i(v)-\varepsilon(1+\frac{3}{\eta})\bar{v}_i(u^i) \ge  U_i(v)-  \frac{4 \varepsilon }{\eta}\bar{v}_i(u^i). \end{equation}
Hence  by using \raf{exp-util}  and \raf{active2 for bv}, we have
\begin{align*}\label{ce1}
\EE[U'_i(v)]&= q_0 \cdot U_i(v) + q_i (\bar{v}_i(u^i)-\eta' L_i)\\
&\le q_0 ( U_i(\bar{v}) + \frac{4 \varepsilon}{\eta} \bar{v}_i(u^i) ) + q_i (\bar{v}_i(u^i)-\eta' L_i)\\
&= \underbrace{q_0  U_i(\bar{v})+q_i (\bar{v}_i(u^i)-\eta' L_i)}_{\EE[U'_i(\bar{v})]}+ q_0 \frac{4 \varepsilon}{\eta} \bar{v}_i(u^i)\\
&= \EE[U'_i(\bar{v})]+q_0 \frac{4 \varepsilon}{\eta}\bar{v}_i(u^i).
\end{align*}
Now applying  \raf{min-util} in the above inequality, we get
\begin{align*}
\EE[U'_i(v)] &\le\EE[U'_i(\bar{v})]+q_0 \frac{4 \varepsilon}{\eta}\bar{v}_i(u^i) \\
&\le\left(1+\frac{q_0}{(1-\bar{\varepsilon})q_i} \frac{4 \varepsilon}{\eta}\right)\EE[U'_i(\bar{v})]\\
 &\le (1+\bar{\varepsilon})\EE[U'_i(\bar{v})], 
\end{align*}
where the last inequality follows from the definition of $\varepsilon$. 
Note that (since $q_0 \le 1$ and $\bar{\varepsilon} \le 1/2$)
\[ \varepsilon \frac{q_0}{(1-\bar{\varepsilon})q_i}\frac{4}{\eta} \le  \varepsilon \frac{1}{(1-\bar{\varepsilon})q_i}\frac{4}{\eta} \le
\frac{\eta \bar{\varepsilon} (1 - q_0)}{8 n} \frac{8}{q_i \eta} = \bar{\varepsilon}. \]

\medskip

\noindent{\it Case 2:} $i\not\in T(v)$. By \raf{min-util}, we have
\[ \EE[U'_i(v)] = q_i \bar v_i(u^i)\le\frac{1}{1-\bar{\varepsilon}}\EE[U'_i(\bar{v})] \le (1 + \varepsilon_0)\EE[U'_i(\bar{v})]  . \]
Since, $\frac{1}{1 - \bar{\varepsilon}} = 1 + \bar{\varepsilon}(1 + \bar{\varepsilon} + \bar{\varepsilon}^2 + \ldots) \le 1 + 2 \bar{\varepsilon} = 1 +\varepsilon_0$.

\noindent{\it Case 3:} $i\in T(v)\setminus T(\bar v)$ and \raf{active2} does not hold for $\bar{v}$. 
Since $U_i(v) \le \bar{v}_i(u^i)$, we have 
\begin{align*}
\EE[U'_i(v)] &=q_0 \cdot U_i(v) + q_i (\bar{v}_i(u^i)-\eta' L_i)
\le (q_0 +q_i )\bar v_i(u^i)- q_i \eta' L_i
<(q_0 + q_i - 1)\bar v_i(u^i)  \\
&\le 0
\le q_i \bar v_i(u^i)
=\EE[U'_i(\bar{v})],
\end{align*}
where the second inequality holds 
because \raf{active2} does not hold for $\bar{v}$ and $q_i \eta'/\eta = 1$.
\smallskip

\noindent{\it Case 4:}  $i\in T(v)\setminus T(\bar v)$ and \raf{active2} holds for $\bar{v}$.
Then \raf{active1} does not hold for $\bar{v}$ and hence 
\begin{equation}\label{not active1 for bv} 
U_i(\bar{v}) =\widehat U_i(\bar{v}) <\frac{q_i}{q_0}\left(\eta' L_i-\bar{\varepsilon}\bar v_i(u^i)\right).
\end{equation}Since \raf{active2} holds for $\bar{v}$, we have \raf{active2 for bv}. Hence by \raf{min-util}, \raf{active2 for bv} and
\raf{not active1 for bv} we have
\begin{align*}
\EE[U'_i(v)]
&=q_0 \cdot U_i(v) + q_i (\bar{v}_i(u^i)-\eta' L_i)\\
&\le q_0 \left(U_i(\bar{v})
+  \frac{4 \varepsilon} {\eta}\bar v_i(u^i)\right) + q_i (\bar v_i(u^i)-\eta' L_i) \\
&\le q_i \eta' L_i- q_i \bar{\varepsilon}\bar v_i(u^i)
+ \frac{4 \varepsilon} {\eta} \bar{v}_i(u^i)+q_i (\bar v_i(u^i)-\eta' L_i) \\
&=(1-\bar{\varepsilon}) q_i \cdot\bar v_i(u^i)+ \frac{4 \varepsilon}{\eta} \bar{v}_i(u^i)\\
&\le\left(1 +\frac{1}{(1-\bar{\varepsilon})q_i } \frac{4 \varepsilon} {\eta}\right)\EE[U'_i(\bar{v})] \\
&\le (1 + \bar{\varepsilon}) \EE[U'_i(\bar{v})],
\end{align*}
where the last inequality follows from the definition of $\varepsilon$ (see Case 1). \smallskip

We finally argue about the approximation ratio. Note that for $i\not\in T(v)$, one of the inequalities \raf{active1} or \raf{active2} does not hold. \Khaled{Also, $U_i(v) \ge 0$ in this case since $p_i=0$,}
and hence $v_i(u^i)<\max\{\eta, \eta'/{\bar{\varepsilon}}\} L_i = \eta' L_i/\bar{\varepsilon} $. 
Since $\mathcal{A}$ returns  allocation $x$ that is $(1 - \varepsilon)$-social efficiency and\footnote{$q_0 n \frac{\eta'}{\bar{\varepsilon}} \le \frac{n \eta}{q_i \bar{\varepsilon}} = \frac{n^2 \bar{\varepsilon} (1 - q_0)^2}{(1 - q_0) \bar{\varepsilon} n^3} = \frac{1 - q_0}{n} = q_i.$} $q_i - q_0 n\frac{\eta'}{\bar{\varepsilon}}\ge 0$, it follows that for any $v\in\cV$, (recall \Khaled{$x=\mathcal{A}(v)$})
\begin{align*}
\EE[v(\Khaled{\widetilde x})]&=q_0 \sum_{i\in T(v)}v_i(x)+ \sum_{i\in[n]}q_i v_i(u^i)\\
&{= q_0\sum_{i\in[n]}v_i(x)- q_0\sum_{i\notin T(v)}v_i(x) + \sum_{i\in[n]} q_i v_i(u^i)}\\
&= q_0v(x)- q_0\sum_{i\notin T(v)}v_i(u^i)   + \sum_{i\in[n]} q_i v_i(u^i)\\
&> q_0v(x)- q_0 \frac{\eta'}{\bar{\varepsilon}}\sum_{i\not\in T(v)}L_i   + \sum_{i\in[n]} q_i v_i(u^i)\\
&= q_0v(x)- q_0 \frac{\eta'}{\bar{\varepsilon}}\sum_{i\not\in T(v)}\sum_{j\neq i}v_j(u^j)   + \sum_{i\in[n]} q_i v_i(u^i)\\
&\ge q_0v(x)- q_0 \frac{\eta'}{\bar{\varepsilon}}n\sum_{{j\in[n]}}v_j(u^j)   + \sum_{i\in[n]} q_i v_i(u^i)\\
&\ge q_0 v(x)+\sum_{i\in[n]}\left(q_i - q_0 n\frac{\eta'}{\bar{\varepsilon}}\right)v_i(u^i)\\
&\ge q_0 (1 - \varepsilon) \cdot \max_{z \in Q} v(z)\\ 
&\ge (1-\varepsilon_0) (1 - \varepsilon) \cdot \max_{z \in Q} v(z).
\end{align*}
\qed

\subsection{Approximately Truthful-in-Expectation Integral Mechanisms}\label{app:3}
In this subsection, we derive a randomized  mechanism $M'$ which returns an integral allocation. Let $\varepsilon > 0$ be arbitrary. First run Algorithm~\ref{rel-TIE} to obtain $x$ and $p(v)$. Then compute a convex decomposition of $\frac{\alpha}{1 +4 \varepsilon}x$, which is
 $ \frac{\alpha}{1 +4\varepsilon}  x = \sum_{j \in \mathcal{N}} \lambda_j^x x^j$.
Finally with probability $\lambda_j^x$ (we use the superscript $x$ to distinguish the convex decompositions of different $x$) return the allocation $x^j$ and charge the $i$-th player the price 
$p_i(v) \frac{v_i(x^j)}{v_i(x)}$, if $v_i(x) > 0$, and  zero otherwise.
We now prove Theorem \ref{conversion}.
\medskip

\noindent{\it Proof of Theorem \ref{conversion}.}
	 Let $M$ be a fractional randomized mechanism obtained in Theorem \ref{approximately truthful mechanism}.
Since $M$ has no positive transfer, $M'$ does neither. 
 $M$ is individually rational with probability $1 - \varepsilon_0$, therefore for any allocation $\bx$,  we have  $\bar{v}_i(\bx) - p_i(\bar{v}) \ge 0$ with probability $1-\varepsilon_0$. 
 So
\[ \bar{v}_i(x^l) - p_i(\bar{v}) \frac{\bar{v}_i(x^l)}{\bar{v}_i(\bx)}  = \left(\bar{v}_i(\bx) - p_i(\bar{v})\right)  \frac{\bar{v}_i(x^l)}{\bar{v}_i(\bx)} \ge 0, \]
hence $M'$ is individually rational with probability $1 - \varepsilon_0$. Now we prove truthfulness. 
Let $\mathit{\EE}[U^{''}_i(\bar{v})]$ be the expected utility of player $i$ when she inputs her true valuation and let $\mathit{\EE}[U^{''}_i(v)]$ be her expected utility when she inputs $v_i$. Then  by definition of $\mathit{\EE}[U^{''}_i(\bar{v})]$, we have
\begin{align*}
\mathit{\EE}[U^{''}_i(\bar{v})]&=
\EE_{\bx \sim M(\bar{v})}\left[\sum_{l \in \mathcal{N}} \lambda^{\bx}_l \left( \bar{v}_i(x^l) - p_i(\bar{v}) \frac{\bar{v}_i(x^l)}{\bar{v}_i(\bx)}\right)\right]\\
&=\EE_{\bx \sim M(\bar{v})}\left[\left( \bar{v}_i(\sum_{l \in \mathcal{N}} \lambda^{\bx}_l x^l) - p_i(\bar{v}) \frac{\bar{v}_i(\sum_{l \in \mathcal{N}} \lambda^{\bx}_l x^l)}{\bar{v}_i(\bx)}\right)\right]\\
 &= \EE_{\bx \sim M(\bar{v})}[ \frac{\alpha}{1 +4 \varepsilon} \bar{v}_i(\bx) - \frac{\alpha}{1 +4 \varepsilon} p_i(\bar{v}) ] \\
  &=  \frac{\alpha}{1 +4 \varepsilon}  \EE_{\bx \sim M(\bar{v})}[ \bar{v}_i(\bx) -  p_i(\bar{v}) ] \\
 &= \frac{\alpha}{1 +4 \varepsilon}  \EE[U'(\bar{v})] \\
&\ge (1 - \varepsilon_0)  \frac{\alpha}{1 +4 \varepsilon}\EE[U'(v)] \\
&= (1 - \varepsilon_0) \frac{\alpha}{1 +4 \varepsilon} \EE_{x\sim M(v)}[\bar{v}(x) - p_i(v)]\\
 &= (1 - \varepsilon_0) \EE_{x\sim M(v)}[\frac{\alpha}{1 +4 \varepsilon} \bar{v}(x) - p_i(v)\frac{\alpha}{1 +4 \varepsilon}\cdot \frac{v_i(x)}{v_i(x)} ]\\
&= (1 - \varepsilon_0) \EE\left[\sum_{l \in \mathcal{N}} \lambda^{x}_l \left( \bar{v}_i(x^l) - p_i(v) \frac{v_i(x^l)}{v_i(x)}\right)\right]\\ 
&= (1 - \varepsilon_0)  \mathit{\EE}[U^{'}_i(v)]. 
\end{align*}
Taking expectation with respect to $x$ shows that the mechanism is $\frac{\alpha(1 - \varepsilon_0)(1 - \varepsilon)}{1+4\varepsilon}$-socially efficient.
\begin{align*}
\mathit{\EE[v(x)]}&=\EE_{x \sim M(v)}\left[\sum_{l \in \mathcal{N}} \lambda^{x}_lv(x^l)\right]
=\EE_{x \sim M(v)}\left[v(\sum_{l \in \mathcal{N}} \lambda^{x}_lx^l)\right] 
=\EE_{x \sim M(v)}\left[v(\frac{\alpha}{1+4\varepsilon}x)\right]\\
&=\frac{\alpha}{1+4\varepsilon}\EE_{x \sim M(v)}[v(x)]
\ge \frac{\alpha}{1+4\varepsilon}(1 - \varepsilon_0)(1 - \varepsilon) \max_{z \in Q} v(z).\\
\end{align*} 
This completes the proof of Theorem~\ref{conversion}.
\qed

\newcommand{\htmladdnormallink}[2]{#1}

\section*{Appendix I: Khandekar's Algorithm for Covering Linear Programs}

Consider a covering linear program:
\begin{align}
\label{cover1}
\min c^{T}x \quad\text{subject to}\qquad Ax\geq b,~~ x\ge 0,
\end{align}
where $A\in\RR_{\ge 0}^{m\times n}$ is an $m\times n$ matrix with non-negative entries and $c\in\RR_{\ge 0}^n$ and  $b\in\RR_{\ge 0}^m$ are non-negative vectors.
 We assume the availability of a $\kappa$-\emph{approximation oracle} for some $\kappa \in (0,1]$. 
\begin{description}
\item[$\cO_\kappa(z)$:] Given $z\in \mathbb{R}^m_{\ge 0}$, the oracle finds a column $j$ of $A$ that maximizes $\frac{1}{c_j}\sum_{i=1}^{m}\frac{z_ia_{ij}}{b_i}$ within a factor of $\kappa$: 
\begin{equation}\label{kappa-approx} \frac{1}{c_j}\sum_{i=1}^{m}\frac{z_ia_{ij}}{b_i} \ge \kappa \cdot \max_{j'\in[n]}\frac{1}{c_{j'}}\sum_{i=1}^{m}\frac{z_ia_{ij'}}{b_i} \end{equation}
\end{description}

We use $A_i$ to denote the $i$-th row of $A$. Algorithm~\ref{Cover-LP} constructs vectors $x(t)\in\RR_{\ge 0}^n$, for $t=0,1,\ldots,$ until $M(t):=\min_{i\in[m]} A_ix(t)/b_i$ becomes at least $T:=\frac{\ln m}{\eps^2}$. Define the \emph{active list} at time $t$ by  $L(t):=\set{i\in[m]}{A_ix(t-1)/b_i <T}$. For $i\in L(t)$, define
\begin{equation}\label{p-cover}
p_{i}(t):=(1-\varepsilon)^{A_i x(t-1)/b_i},
\end{equation}
and set $p_i(t)=0$ for $i\not \in L(t)$. 
At each time $t$, the algorithm calls the oracle with the vector 
$z_t=p(t)/\onenorm{p(t)}$, and increases the variable $x_{j(t)}$ by 
\begin{equation}\label{def of delta}
\delta(t):=\min_{i\in L(t) \text{ and }a_{i,j(t)} \not= 0\ \ \ } \frac{b_i}{a_{i,j(t)}},
\end{equation}
where $j(t)$ is the index returned by the oracle. 

\begin{algorithm}[t]
\caption{Covering$(\cO_\kappa)$}
\label{Cover-LP}
\begin{algorithmic}[1]
\REQUIRE a covering system $(A,b,c)$ given by a $\kappa-$approximation oracle $\cO_\kappa$, where $A\in \mathbb{R}_{\ge 0}^{m\times n},$ $b\in \mathbb{R}_{> 0}^m$, $c\in\RR^n_{> 0}$, and an accuracy parameter $\varepsilon\in(0,1/2]$ 
\ENSURE A feasible solution $\hat x\in\RR_{\ge 0}^n$ to \raf{cover} s.t. $c^T\hat x\le\frac{(1+4\varepsilon)}{\kappa}z^*$
  \STATE $x(0):=0$; $t:=0$; and $T:=\frac{\ln m}{\varepsilon^2}$
  \WHILE{$M(t)<T$}
    \STATE $t:=t+1$
    \STATE Let $j(t):=\cO_\kappa(p(t)/\onenorm{p(t)})$  
    \STATE $x_{j(t)}(t):=x_{j(t)}(t-1)+\delta(t)$ and $x_j(t) = x_j(t-1)$ for $j \not= j(t)$ \label{sss2-cover}
  \ENDWHILE
  \RETURN $\hat x=\frac{x(t)}{M(t)}$
\end{algorithmic}
\end{algorithm}

\paragraph{\bf Proof  of Theorem \ref{t1}.} Note that the RHS of \raf{kappa-approx} is positive for our choice of $z_t$ since every row of $A$ contains a non-zero entry and hence $\sum_{i \in L(t)} p_i(t) a_{ij(t)}/(b_i c_{j(t)}) > 0$. This conclude that there exist  at least one $i\in L(t)$ which $a_{i,j(t)}$ is non zero and
  thus $\delta(t) > 0$ always. In each iteration, some entry of $x$ is increased and hence the values $A_i x(t)/b_i$ are non-decreasing. Thus $L(t+1) \subseteq L(t)$ for all $t$. At the end, we scale $x(t)$ by $M(t)$ to guarantee feasibility.


Let $\bone_j$ denote the $j$-th unit vector of dimension $n$ and $B\in\RR^{m\times m}$ be a diagonal matrix  with entries $b_{ii}=b_i$. Feasibility is obvious since we scale by $M(t)$. The bound on the number of iterations is also obvious since in each iteration at least one of the $A_ix/b_i$ increases by one and we remove $i$ from the active list once $A_ix/b_i$ reaches $T$.  We conclude that the number of iterations is bounded by $m \ceil{T}$. Let $t_0$ be the number of iterations, i.e., vectors $x(0)$, $x(1)$, \ldots, $x(t_0)$ are defined and
$M(t_0 - 1) < T \le M(t_0)$. In the $t$-th iteration exactly one entry of $x$ is increased by $\delta(t)$ and hence $\bone^T x(t_0) = \sum_{1 \le t \le t_0} \delta(t)$ and $A_i x(t)/b_i \le A_i x(t - 1)/b_i+ 1$ for $i \in L(t)$. To show \raf{bd11}, we analyze the decrease of $\|p(t)\|_1$. Let 
$t \le t_0$. Then
\begin{align}\label{bd2}
 \sum_{i\in L(t)} (1-\varepsilon)^{A_i x(t)/b_i}\nonumber
&=\sum_{i\in L(t)} (1-\varepsilon)^{A_ix(t-1)/b_i+\delta(t)A_i\bone_{j(t)}/b_i}\nonumber\\
&=\sum_{i\in L(t)} p_i(t) (1-\varepsilon)^{\delta(t)A_i\bone_{j(t)}/b_i} \nonumber\\
& \le \sum_{i\in L(t)} p_i(t) (1-\varepsilon\delta(t)A_i\bone_{j(t)}/b_i)\nonumber\\
& \text{(using \raf{def of delta}  conclude that  $\delta(t)A_i\bone_{j(t)}/b_i\le1$ and }\nonumber\\
& \text{$(1-\varepsilon)^x\leq 1-\varepsilon x$ for all $\varepsilon\in [0,1),~x\in [0,1]$)}\nonumber\\
&=\|p(t)\|_1\left(1-\frac{\varepsilon\delta(t)p(t)^T B^{-1}A\bone_{j(t)}}{\|p(t)\|_1}\right)\nonumber\\
&\le \|p(t)\|_1 e^{-\varepsilon\delta(t)\frac{p(t)^T}{\|p(t)\|_1} B^{-1}A\bone_{j(t)}} \quad \text{since $1-x\leq e^{-x}$.}
\end{align}
By using $L(t+1)\subseteq L(t)$, we have
\begin{align}\label{bd23}
\| p(t+1) \|_1 \nonumber
&= \sum_{i \in L_{(t+1)}} (1 - \varepsilon)^{A_i x(t)/b_i}\nonumber\\
& \le \sum_{i \in L_{(t)}} (1 - \varepsilon)^{A_i x(t)/b_i}
\end{align}
and hence applying inequalities \raf{bd2} and \raf{bd23} we get,
\begin{equation}\label{bd2a}
\|p(t+1)\|_1 \le \sum_{i \in L(t)} (1 - \varepsilon)^{A_i x(t)/b_i} \le \|p(t)\|_1 e^{-\varepsilon\delta(t)\frac{p(t)^T}{\|p(t)\|_1} B^{-1}A\bone_{j(t)}}.  \end{equation}
Let $i_0 \in L(t_0)$ be arbitrary. Then 
\begin{align*}
(1 - \varepsilon)^{A_{i_0}x(t_0)/b_{i_0}} &\le \sum_{i\in L(t_0)} (1-\varepsilon)^{A_i x(t_0)/b_{i_0}}\\
&\le \|p(t_0)\|_1 e^{-\varepsilon\delta(t_0)\frac{p(t_0)^T}{\|p(t_0)\|_1} B^{-1}A\bone_{j(t_0)}}\\
&\le \|p(0)\|_1e^{-\varepsilon\sum_{1 \le t \le t_0} \delta(t)\frac{p(t)^T}{\|p(t)\|_1} B^{-1}A\bone_{j(t)}},
\end{align*}
where the second inequality uses (\ref{bd2}) for $t = t_0$ and the third inequality uses (\ref{bd2a}) for $0 \le t < t_0$.
Taking logs and using $\|p(0)\|_1=m$, we conclude that
\begin{equation}\label{e11}
A_{i_0}x(t_0)/b_{i_0}\cdot\ln(1-\varepsilon)\le \ln m-\varepsilon\sum_{1 \le t \le t_0}\delta(t)\frac{p(t)^T}{\|p(t)\|_1} B^{-1}A\bone_{j(t)} 
\end{equation}
We next relate the  objective value $c^Tx(t_0)=\sum_{1 \le t \le t_0}c_{j(t)} \delta(t)$ at time $t_0$ to the optimal value $z^*$ by the following claim.
\begin{claim}\label{cl11}
$\sum_{1 \le t \le t_0}\delta(t)\frac{p(t)^T}{\|p(t)\|_1} B^{-1}A\bone_{j(t)}\ge\frac{\kappa\cdot c^Tx(t_0)}{z^*}.$
\end{claim}
\begin{proof}
Let $x^*\in\RR_{\ge 0}^n$ be an optimal solution to \raf{cover}. Since $x^*$ is feasible,  $B^{-1}Ax^*\ge \bone$, and thus for any $t$,  
$$ p(t)^T B^{-1}Ax^*\ge p(t)^T \bone = \|p(t)\|_1.$$
By the choice of the index $j(t)$, we have that $ \frac{1}{c_{j(t)}}p(t)^TB^{-1}A\bone_{j(t)}\ge  \frac{1}{c_j}\kappa p(t)^T B^{-1}A\bone_j$ for all $j\in[n]$. Since $z^* = c^T x^*$, we conclude further
\begin{align*}
z^* p(t)^T B^{-1}A\bone_{j(t)}&=\sum_{j\in[n]}c_jx^*_j p(t)^T B^{-1}A\bone_{j(t)}\\
&=\sum_{j\in[n]}c_j\cdot\frac{c_{j(t)}}{c_{j(t)}}x^*_j p(t )^T B^{-1}A\bone_{j(t)}\\
&\ge\sum_{j\in[n]}c_j\cdot\frac{c_{j(t)}}{c_j}x^*_j\kappa p(t )^T B^{-1}A\bone_{j}\\
&=\kappa c_{j(t)}p(t)^T B^{-1}Ax^*\\
&\ge \kappa c_{j(t)} \| p(t) \|_1.
\end{align*}
Multiplying both sides of this inequality by $\delta(t)/\onenorm{p(t)}$ and summing up over $1\le t\le t_0$ finishes the proof of the claim.
\qed
\end{proof}

\noindent
Using  the claim above, we deduce from \raf{e11} 
\[ A_{i_0} x(t_0)/b_{i_0}\cdot\ln(1-\varepsilon)\le \ln m-\varepsilon\cdot\frac{\kappa\cdot c^Tx(t_0)}{z^*} \]
Dividing both sides by $M(t_0)$, arranging, and using $M(t_0)\ge T = (\ln m)/\varepsilon^2$, $A_{i_0} x(t_0)/b_{i_0} \le A_{i_0} x(t_0 - 1)/b_{i_0} + 1 \le T + 1$,  and $\frac{\ln\frac{1}{1-\varepsilon}}{\varepsilon}\leq 1+2\varepsilon$, valid for all $\varepsilon\in (0,\frac{1}{2}]$, we obtain
\begin{align*}
\frac{\kappa\cdot c^T\hat x}{z^*} =  \frac{\kappa\cdot c^Tx(t_0)}{M(t_0)z^*} 
&\le \frac{\ln\frac{1}{1-\varepsilon}}{\varepsilon}\cdot\frac{A_{i_0}x(t_0)/b_{i_0}}{M(t_0)}+\frac{\ln m}{\varepsilon\cdot M(t_0)}\\
&\le  (1 + 2 \varepsilon) \frac{T + 1}{T}   + \varepsilon \le 1 + 4\varepsilon.
\end{align*}
\qed

\paragraph{\bf Proof  of  Corollary~\ref{c1}.}
Recall~\raf{e11}:
\[
A_{i_0}x(t_0)/b_{i_0}\cdot\ln(1-\varepsilon)\le \ln m-\varepsilon\sum_{1 \le t \le t_0}\delta(t)\frac{p(t)^T}{\|p(t)\|_1} B^{-1}A\bone_{j(t)} .\]
With assumption $b=\bone$, we have,
\[
A_{i_0}x(t_0)\cdot\ln(1-\varepsilon)\le \ln m-\varepsilon\sum_{1 \le t \le t_0}\delta(t)\frac{p(t)^T}{\|p(t)\|_1} A\bone_{j(t)} .\]

The vector $z_t = p(t)/\onenorm{p(t)}$ satisfies $\bone^T z_t = 1$. Apply oracle $\cO'$ with input vector $z_t$, it returns index $j(t)$  such that we have $\frac{p(t)^T}{\|p(t)\|_1} A\bone_{j(t)} \ge 1$. Thus, we have 
\[ A_{i_0} x(t_0)\cdot\ln(1-\varepsilon)\le \ln m-\varepsilon\cdot\bone^Tx(t_0).\]
Proceeding as in the proof of Theorem~\ref{t1}, we get the result. 
\qed

\end{document}